\documentclass[draftcls,12pt,journal]{IEEEtran}
\usepackage{amsmath,color,amssymb,graphicx,epsfig,cite,psfrag,enumerate}

\newtheorem{Lemma}{\it Lemma}
\newtheorem{Corollary}{\it Corollary}
\newtheorem{Definition}{\it Definition}
\newtheorem{Remark}{\it Remark}
\newtheorem{Theorem}{\it Theorem}

\def\bq{{\mathbf{q}}}
\def\bw{{\mathbf{w}}}
\def\bW{{\mathbf{W}}}
\def\bV{{\mathbf{V}}}
\def\bs{{\mathbf{s}}}
\def\bx{{\mathbf{x}}}
\def\by{{\mathbf{y}}}
\def\bz{{\mathbf{z}}}
\def\bu{{\mathbf{u}}}
\def\bv{{\mathbf{v}}}

\def\bY{{\mathbf{Y}}}
\def\bQ{{\mathbf{Q}}}
\def\bS{{\mathbf{S}}}
\def\bU{{\mathbf{U}}}
\def\sptfirst{\text{sp1}}
\def\spttwo{\text{sp2}}
\def\sim{\text{sim}}
\def\suc{\text{suc}}
\def\inn{\text{in}}
\begin{document}
\onecolumn
\title{\huge On the Achievable Rate Regions for Interference Channels with Degraded Message Sets$^{\dagger}$}
\author{\authorblockN{Jinhua Jiang and Yan Xin} \\
\authorblockA{Department of Electrical and Computer Engineering \\
National University of Singapore, Singapore 117576 \\
Email: \{ jinhua.jiang, elexy \}@nus.edu.sg}
\thanks{$^{\dagger}$ The work is supported by the
National University of Singapore (NUS) under start-up grants
R-263-000-314-101 and R-263-000-314-112 and by a NUS Research
Scholarship. The correspondence author of the paper is Dr. Yan Xin
(tel. no. +65 6516-5513 and fax no. +65 6779-1103).}}
\markboth{IEEE Transactions on Information Theory (submitted)}
{Jiang and Xin}
\renewcommand{\thepage}{}
\markboth{}{} \maketitle
\begin{abstract}
The interference channel with degraded message sets (IC-DMS)
refers to a communication model in which two senders attempt to
communicate with their respective receivers simultaneously through
a common medium, and one of the senders has {\it complete} and
{\it a priori} (non-causal) knowledge about the message being
transmitted by the other. A coding scheme that collectively has
advantages of cooperative coding, collaborative coding, and dirty
paper coding, is developed for such a channel. With resorting to
this coding scheme, achievable rate regions of the IC-DMS in both
discrete memoryless and Gaussian cases are derived, which, in
general, include several previously known rate regions. Numerical
examples for the Gaussian case demonstrate that in the {\it
high-interference-gain} regime, the derived achievable rate
regions offer considerable improvements over these existing
results.
\end{abstract}

\newpage
\pagenumbering{arabic} \markboth{IEEE Transactions on Information
Theory (submitted)} {Jiang and Xin}
\section{Introduction}

The interference channel with degraded message sets (IC-DMS)
refers a communication model in which two senders attempt to
communicate with their respective receivers simultaneously through
a common medium, and one of the senders has {\it complete} and
{\it a priori} (non-causal) knowledge about the message being
transmitted by the other. Such a model generically characterizes
some realistic communication scenarios taking place in cognitive
radio channels \cite{Tarokh06:ic_dms_cog,jovicic06:cog_ICDMS} or
in wireless sensor networks over a correlated field
\cite{wuwei06_icdms,Kramer06_UCSDworkshop:IFC_uniDcoop}, which we
illustrate in Figs. \ref{fig_scenarios}(a) and
\ref{fig_scenarios}(b).

\begin{figure}[b]
\centering \psfrag{w1}{$w_1$} \psfrag{w2}{$w_2$}
\includegraphics{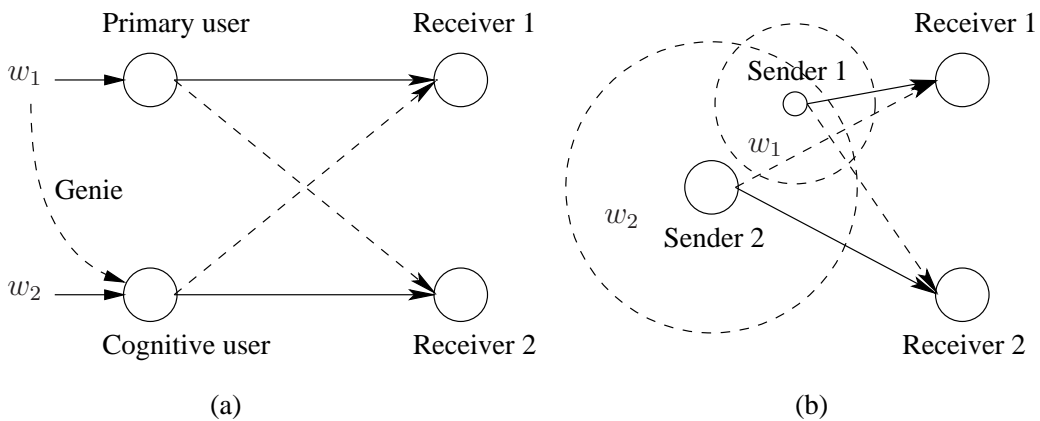}
\caption{(a) A genie-aided cognitive radio channel
\cite{Tarokh06:ic_dms_cog}, in which the Genie informs the
cognitive user of what the primary user will transmit; (b) A
four-node wireless sensor network \cite{wuwei06_icdms}, in which
sender 2 senses a larger area such that it knows what information
sender 1 obtains.} \label{fig_scenarios}
\end{figure}

%
From an information-theoretic perspective, the IC-DMS have been
investigated in
\cite{Tarokh06:ic_dms_cog,Kramer06_UCSDworkshop:IFC_uniDcoop,jovicic06:cog_ICDMS,wuwei06_icdms}.
Specifically, several achievable rate results have been obtained
in
\cite{Tarokh06:ic_dms_cog,Kramer06_UCSDworkshop:IFC_uniDcoop,jovicic06:cog_ICDMS,wuwei06_icdms},
and the capacity regions for two special cases have been
characterized in
\cite{Kramer06_UCSDworkshop:IFC_uniDcoop,jovicic06:cog_ICDMS,wuwei06_icdms}.
The main achievable rate region in \cite{Tarokh06:ic_dms_cog} was
obtained by incorporating the Gel'fand-Pinsker coding
\cite{gelfand_pinsker80:channel_random_param} into the well-known
coding scheme applied to the interference channel (IC)
\cite{Carleial78:IFC,Han81:IFC}. In this coding scheme, each of
the two senders splits its message into two sub-messages, and
allows its non-pairing receiver to decode one of the sub-messages.
Knowing the two sub-messages and the corresponding codewords which
sender 1 wishes to transmit, sender 2 applies the Gel'fand-Pinsker
coding to encode its own sub-messages by treating the codewords of
sender 1 as known interferences. It has been also shown in
\cite[Corollary 2]{Tarokh06:ic_dms_cog} that, an improved
achievable rate region can be attained by time-sharing between the
early derived rate region and a so called fully-cooperative rate
point achieved by letting sender 2 use all its power to transmit
sender 1' messages. A different coding scheme was adopted in
\cite{jovicic06:cog_ICDMS} and \cite{wuwei06_icdms}, in which
neither of the senders splits its message into two sub-messages,
and receiver 2 does not decode any transmitted information from
sender 1. Since sender 2 knows what sender 1 wishes to transmit,
sender 2 is allowed to: 1) apply the Gel'fand-Pinsker coding to
encode its own message; and 2) partially cooperate with sender 1
using superposition coding. It has been proven in
\cite{jovicic06:cog_ICDMS, wuwei06_icdms} that, this is the
capacity-achieving scheme for the Gaussian IC-DMS in the {\it
low-interference-gain} regime, in which the normalized link gain
between sender 2 and receiver 1 is less than or equal to $1$.
\begin{figure}[t]
\centering \psfrag{w1}{$w_1$} \psfrag{w2}{$w_2$}
\includegraphics{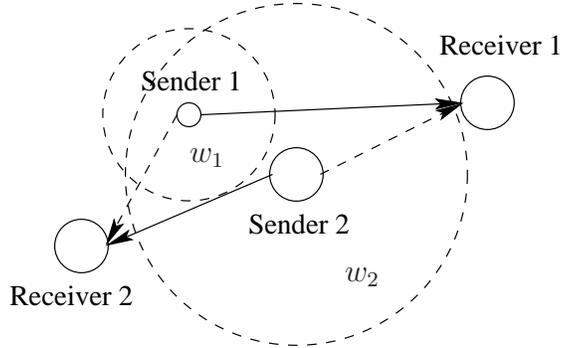}
\caption{An interference channel with degraded message sets in
which sender 2 is close to receiver 1.} \label{fig_scenario_our}
\end{figure}

However, in practice, due to the mobility of the users or random
distributions of the sensors, sender 2 may be geographically
located near to receiver 1, as illustrated in Fig.
\ref{fig_scenario_our}. It is likely, in such a situation, that
the Gaussian IC-DMS is in the {\it high-interference-gain} regime,
in which the normalized link gain between sender 2 and receiver 1
is greater than 1. In fact, the findings in this paper reveal that
the achievable rate region, which was proven to be the capacity
region in the low-interference-gain regime in
\cite{jovicic06:cog_ICDMS} and \cite{wuwei06_icdms}, is {\it
strictly} non-optimal for the Gaussian IC-DMS in the
high-interference-gain regime.

In this paper, we develop a new coding scheme for the IC-DMS to
improve existing achievable rate regions. Our coding scheme
differs from one proposed in
\cite{jovicic06:cog_ICDMS,wuwei06_icdms} in the way that, sender 2
splits its message into two sub-messages, and encodes both
sub-messages using Gel'fand-Pinsker coding. Moreover, receiver 1
is required to jointly decode the message from sender 1 and one
sub-message from sender 2. With this coding scheme, we derive our
main achievable rate region for the discrete memoryless case. For
comparison purpose, we compromise either the coding flexibility
(fixing an auxiliary random variable as a constant), or the
advantage of simultaneous decoding \cite{Han81:IFC}, to obtain two
subregions of the main achievable rate region. The obtained
subregions are shown to either include or be the same as the
existing ones. We further extend the obtained regions from the
discrete memoryless case to the Gaussian case, and show by
numerical examples that our Gaussian achievable rate results
strictly improve the existing ones in the high-interference-gain
regime.

The rest of the paper is organized as follows. In Section
\ref{channel_model}, we introduce the channel model of the IC-DMS,
and the related terminologies. In Section
\ref{section_region_general_MOST}, we present the main achievable
result for the discrete memoryless case with a detailed proof. In
Section \ref{section_region_general}, we derive two subregions of
the main achievable rate region, and we show that the derived
subregions include several existing results as special cases.
Lastly, in Section \ref{section_region_Gaussian}, we extend our
rate regions from the discrete memoryless case to the Gaussian
case, and compare them with the existing results.

{\it Notations:} Random variables and their realizations are
denoted by upper case letters and lower case letters respectively,
e.g., $X$ and $x$. Bold lower (upper) case letters are used to
denote vectors (matrices), e.g., $\mathbf{x}$ and
$\mathbf{\Sigma}$. Calligraphic fonts are used to denote sets,
e.g., $\mathcal{X}$ and $\mathcal{R}$.

\section{The Channel Model}\label{channel_model}

Consider the IC-DMS (also termed as the genie-aided cognitive
radio channel in \cite{Tarokh06:ic_dms_cog}) depicted in Fig.
\ref{fig_ICDMS}, in which sender 1 wishes to transmit a message
(or message index), $w_1 \in \mathcal{M}_1 := \{1,...,M_1\}$, to
receiver 1 and sender 2 wishes to transmit its message, $w_2
\in\mathcal{M}_2 := \{1,...,M_2\}$, to receiver 2. Typically, this
{\it discrete memoryless} IC-DMS is described by a tuple
$(\mathcal{X}_1,\mathcal{X}_2,\mathcal{Y}_1,\mathcal{Y}_2,p(y_1,y_2|x_1,x_2))$,
where $\mathcal{X}_1$ and $\mathcal{X}_2$ are the channel input
alphabets, $\mathcal{Y}_1$ and $\mathcal{Y}_2$ are the channel
output alphabets, and $p(y_1,y_2|x_1,x_2)$ denotes the conditional
probability of $(y_1,y_2) \in \mathcal{Y}_1\times\mathcal{Y}_2$
given $(x_1,x_2) \in \mathcal{X}_1\times\mathcal{X}_2$. The
channel is discrete memoryless in the sense that
\begin{align}
  p(y_{1,t},y_{2,t}|x_{1,t},x_{2,t},x_{1,{t-1}},x_{2,{t-1}},...) =
  p(y_{1,t},y_{2,t}|x_{1,t},x_{2,t}),
\end{align}
for every discrete time instant $t$ in a synchronous transmission.
In terms of the channel input-output relationship, the IC-DMS is
the same as the IC. However, in the IC-DMS, sender 2 is able to
noncausally obtain the knowledge of the message $w_1$, which will
be transmitted from sender 1. This is the key difference between
the IC-DMS and IC in terms of the information flow.  We next
present the following standard definitions with regard to the
existence of codes and achievable rates for the discrete
memoryless IC-DMS channel.

\begin{figure}[t]
\centering \psfrag{w1}{$w_1$} \psfrag{w2}{$w_2$}
\psfrag{x1}{$x^n_1$} \psfrag{x2}{$x^n_2$} \psfrag{y1}{$y^n_1$}
\psfrag{y2}{$y^n_2$} \psfrag{w1'}{$\hat{w}_1$}
\psfrag{w2'}{$\hat{w}_2$} \psfrag{P}{$p(y_1,y_2|x_1,x_2)$}
\psfrag{f_1}{$f_1(\cdot)$} \psfrag{f_2}{$f_2(\cdot)$}
\psfrag{g_1}{$g_1(\cdot)$} \psfrag{g_2}{$g_2(\cdot)$}
\includegraphics{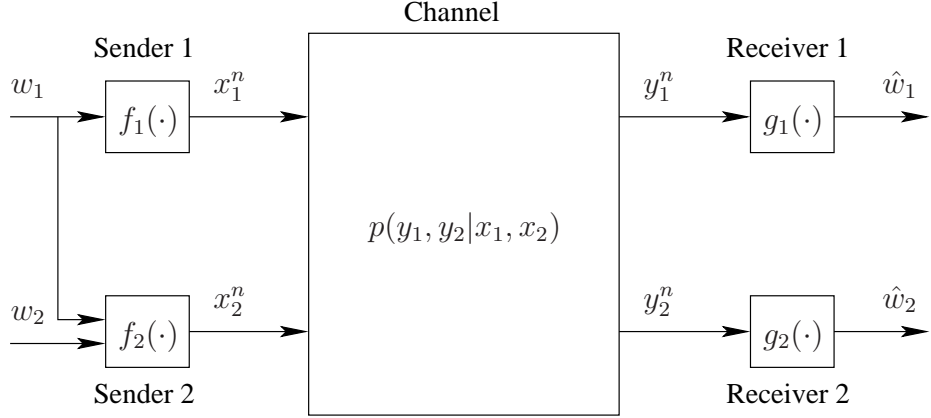}
\caption{An interference channel with degraded message sets.}
\label{fig_ICDMS}
\end{figure}

\begin{Definition}
An $(M_1, M_2, n, P_e)$ code exists for the discrete memoryless
IC-DMS, if and only if there exist two encoding functions
\[ f_1:
\mathcal{M}_1 \rightarrow \mathcal{X}^n_1, \quad f_2:
\mathcal{M}_1 \times \mathcal{M}_2 \rightarrow \mathcal{X}^n_2,
\]
and two decoding functions
\[
g_1: \mathcal{Y}^n_1 \rightarrow {\mathcal{M}}_1,\quad g_2:
\mathcal{Y}^n_2 \rightarrow {\mathcal{M}}_2,
\] such that
$\max\{P_{e,1}^{(n)}, P_{e,2}^{(n)}\} \leq P_e$, where
$P_{e,1}^{(n)}$ and  $P_{e,2}^{(n)}$ denote the respective average
probabilities of error at decoders $1$ and $2$, and are computed
as
\begin{align*}
&P^{(n)}_{e,1} = \frac{1}{M_1M_2}\sum_{w_1w_2} p(\hat{w}_1\neq w_1 |(w_1,w_2) \text{~were~sent}), \\
&P^{(n)}_{e,2} = \frac{1}{M_1M_2}\sum_{w_1w_2} p(\hat{w}_2\neq w_2
|(w_1,w_2) \text{~were~sent}).
\end{align*}
\end{Definition}
\vspace*{0.5cm}
\begin{Definition} A non-negative rate pair $(R_1, R_2)$ is achievable for the
IC-DMS, if for any given $0<P_e<1$ and any sufficiently large $n$,
there exists a $(2^{nR_1}, 2^{nR_2}, n, P_e)$ code for the
channel. The capacity region of the IC-DMS is the set of all the
achievable rate pairs for the channel, and an achievable rate
region is a subset of the capacity region.
\end{Definition}

It should be noted that from an information-theoretic standpoint,
the IC can not be simply treated as a special case of the IC-DMS
in the sense that the capacity region of the IC-DMS, if any, does
not imply a capacity region of the IC.

\section{An Achievable Rate Region for the
Discrete Memoryless IC-DMS}\label{section_region_general_MOST}

In this section, we present the main achievable rate region for
the discrete memoryless IC-DMS, which is the primary result in
this paper.

Consider auxiliary random variables $W$, $U$, $\tilde{U}$, $V$,
$\tilde{V}$ and a time-sharing random variable $Q$, defined on
arbitrary finite sets $\mathcal{W}$, $\mathcal{U}$,
$\tilde{\mathcal{U}}$, $\mathcal{V}$, $\tilde{\mathcal{V}}$ and
$\mathcal{Q}$ respectively. Let $\mathcal{P}$ denote the set of
all joint probability distributions $p(\cdot)$ that factor in the
form of
\begin{align}
  p(q,w,x_1,u, \tilde{u},v, \tilde{v},x_2,y_1,y_2) =&
  p(q)p(w,x_1|q)p(u,\tilde{u}|w,q)p(v,\tilde{v}|w,q)\nonumber\\
  &\cdot p(x_2|\tilde{u},\tilde{v},w,q)p(y_1,y_2|x_1,x_2), \label{joint_pdf_most}
\end{align}
where $w$, $u$, $\tilde{u}$, $v$, $\tilde{v}$, and $q$ are
realizations of random variables $W$, $U$, $\tilde{U}$, $V$,
$\tilde{V}$ and $Q$.

Let $\mathcal{R}(p)$ denote the set of all non-negative rate pairs
$(R_1, R_2)$ such that the following inequalities hold
simultaneously
\begin{align}
R_1 &\leq I(W;Y_1U|Q), \label{region_ineq1}\\
R_2 &\leq I(UV;Y_2|Q) - I(U;W|Q)-I(V;W|Q),\\
R_1 + R_2 &\leq I(UW;Y_1|Q)+ I(V;Y_2U|Q)-I(U;W|Q)- I (V;W|Q);
\label{region_ineq5}
\\
0& \leq I(UW;Y_1|Q) - I(U;W|Q) , \label{region_cons_1}\\
0&\leq I(U;Y_2V|Q) - I(U;W|Q) ,
\\
0&\leq I(V;Y_2U|Q) - I(V;W|Q) ,
\\
0&\leq I(UV;Y_2|Q) - I(U;W|Q)-I(V;W|Q)\label{region_cons_4},
\end{align}
for a given joint distribution $p(\cdot) \in \mathcal{P}$.

Let $\mathcal{C}$ denote the capacity region of the discrete
memoryless IC-DMS, and let
\[
\mathcal{R} = \bigcup_{p(\cdot)\in \mathcal{P}}\mathcal{R}(p).
\]

\begin{Theorem}
The region $\mathcal{R}$ is an achievable rate region for the
discrete memoryless IC-DMS, i.e., $\mathcal{R}\subseteq
\mathcal{C}$. \label{theorem_region_most}
\end{Theorem}

\begin{proof}
Before presenting the proof of Theorem \ref{theorem_region_most},
we state the following lemma as it will be frequently used in the
proof.

\begin{Lemma}[{\cite[Theorem 14.2.3]{cover_IT_book}}]
Let $A^{(n)}_{\epsilon}$ denote the typical set for the
probability mass distribution $p(s_1,s_2,s_3)$, and let
$P(\bS_1'=\bs_1, \bS_2'=\bs_2, \bS_3'=\bs_3) = \prod_{i=1}^n
p(s_{1i}|s_{3i})p(s_{2i}|s_{3i})p(s_{3i})$, then $P\{(\bS_1',
\bS_2', \bS_3')\in A^{(n)}_{\epsilon}\} \doteq
2^{-n(I(S'_1;S'_2|S'_3)\pm 6\epsilon)}$.
\label{lemma_cover}\end{Lemma}

To prove this theorem we apply the notion of the asymptotic
equipartition property (APE) \cite{cover_IT_book}. Our coding
scheme is mainly based on the arguments of superposition coding
\cite{cover75:broadcast} and Gel'fand-Pinsker coding
\cite{gelfand_pinsker80:channel_random_param}. Specifically,
sender 1 independently encodes its message $w_1$ as a whole; while
sender 2 needs split its message into two parts, i.e., $w_2 =
(w_{21},w_{22})$, and encode them separately. Both $w_{21}$ and
$w_{22}$ are encoded using the Gel'fand-Pinsker approach, but they
are processed differently at the receivers. The message $w_{22}$
will be decoded by receiver 2 only, while $w_{21}$ will be decoded
by both receivers. Moreover, knowing the message and codeword
which sender 1 is going to transmit, sender 2 not only can apply
Gel'fand-Pinsker coding to deal with the known interference, but
also can cooperate with sender 1 to transmit $w_1$ using
superposition coding. Let $R_{21}$ and $R_{22}$ denote the rates
of $w_{21}$ and $w_{22}$ respectively, i.e.,
$w_{21}\in\{1,\ldots,2^{nR_{21}}\}$ and
$w_{22}\in\{1,\ldots,2^{nR_{22}}\}$. If receiver 1 can decode
$w_1$ and receiver 2 can decode both $w_{21}$ and $w_{22}$ with
vanishing probabilities of error, then $(R_1, R_{21}+R_{22})$ is
an achievable rate pair for the IC-DMS.

To prove that the entire region $\mathcal{R}$ is achievable for
the channel, it is sufficient to prove that $\mathcal{R}(p)$ is
achievable for a fixed joint probability distribution $p(\cdot)\in
\mathcal{P}$.

\subsection{Random Codebook Generation}
Consider a fixed joint distribution $p(\cdot)\in \mathcal{P}$, and
a random time-sharing codeword $\bq$ of length $n$, which is given
to both senders and receivers. The codeword $\bq$ is assumed to be
generated according to $\prod^n_{i=1}p(q_{i})$.

Generate $2^{nR_1}$ independent codewords $\bw(j)$,
$j\in\{1,\ldots,2^{nR_{1}}\}$, according to
$\prod^n_{i=1}p(w_{i}|q_i)$; and for each $\bw(j)$ generate one
$\bx_1(j)$, according to $\prod^n_{i=1}p(x_{1i}|w_{i}q_i)$.
Similarly, generate $2^{n\tilde{R}_{21}}$ independent codewords
$\bu(l_1)$, $l_1\in\{1,\ldots,2^{n\tilde{R}_{21}}\}$, according to
$\prod^n_{i=1}p(u_{i}|q_i)$, and generate $2^{n\tilde{R}_{22}}$
independent codewords $\bv(l_2)$,
$l_2\in\{1,\ldots,2^{n\tilde{R}_{22}}\}$, according to
$\prod^n_{i=1}p(v_{i}|q_i)$.

For each codeword pair $(\bu(l_1),\bw(j))$, generate one codeword
$\tilde{\bu}(l_1,j)$ according to
$\prod^n_{i=1}p(\tilde{u}_{i}|u_{i}(l_1)w_i(j)q_i)$, and similarly
for each codeword pair $(\bv(l_2),\bw(j))$, generate one codeword
$\tilde{\bv}(l_2,j)$ according to
$\prod^n_{i=1}p(\tilde{v}_{i}|v_{i}(l_2)w_i(j)q_i)$. Lastly, for
each codeword triple $(\bu(l_1),\bv(l_2),\bw(j))$, generate one
codeword $\bx_2(l_1,l_2,j)$ according to
$\prod^n_{i=1}p(x_{2i}|\tilde{u}_{i}(l_1)\tilde{v}_{i}(l_2)w_i(j)q_i)$.

Now uniformly distribute $2^{n\tilde{R}_{21}}$ codewords
$\bu(l_1)$ into $2^{nR_{21}}$ bins indexed by
$k_1\in\{1,\ldots,2^{nR_{21}}\}$ such that each bin contains
$2^{n(\tilde{R}_{21}-R_{21})}$ codewords; uniformly distribute
$2^{n\tilde{R}_{22}}$ codewords $\bv(l_2)$ into $2^{nR_{22}}$ bins
indexed by $k_2\in\{1,\ldots,2^{nR_{22}}\}$ such that each bin
contains $2^{n(\tilde{R}_{22}-R_{22})}$ codewords.

The entire codebook is revealed to both senders and receivers.

\subsection{Encoding and Transmission}
We assume that the senders want to transmit a message vector
$(w_1,w_{21},w_{22}) = (j, k_1, k_2)$. Sender 1 simply encodes the
message as codeword $\bx_1(j)$ and sends the codeword with $n$
channel uses. Sender 2 will first need to look for a codeword
$\bu(\hat{l}_1)$ in bin $k_1$ such that $(\bu(\hat{l}_1), \bw(j),
\bq)\in A_{\epsilon}^{(n)}$, and a codeword $\bv(\hat{l}_2)$ in
bin $k_2$ such that $(\bv(\hat{l}_2), \bw(j), \bq)\in
A_{\epsilon}^{(n)}$. If sender 2 fails to do so, it will randomly
pick a codeword $\bu(\hat{l}_1)$ from bin $k_1$ or a codeword
$\bv(\hat{l}_2)$ from bin $k_2$. Sender 2 then transmits codeword
$\bx_2(\hat{l}_1,\hat{l}_2,j)$ through $n$ channel uses. We
further assume that the transmissions are perfectly synchronized.
\subsection{Decoding}
Receiver 1 first looks for all the index pairs
$(\hat{j},\hat{\hat{l}}_1)$ such that $(\bw(\hat{j}),
\bu(\hat{\hat{l}}_1), \by_1, \bq)\in A_{\epsilon}^{(n)}$. If
$\hat{j}$ in all the index pairs found are the same, receiver 1
determines $w_1=\hat{j}$, otherwise declares an error.

Receiver 2 will first look for all index pairs
$(\bar{\hat{l}}_1,\hat{\hat{l}}_2)$ such that $(
\bu(\bar{\hat{l}}_1), \bv(\hat{\hat{l}}_2),\by_2, \bq)\in
A_{\epsilon}^{(n)}$. If $\bar{\hat{l}}_1$ in all the index pairs
found are indices of codewords $\bu(\bar{\hat{l}}_1)$ from the
same bin with index $\hat{k}_1$, and $\hat{\hat{l}}_2$ in all the
index pairs found are indices of codewords $\bv(\hat{\hat{l}}_2)$
from the same bin with index $\hat{k}_2$, then receiver 2 will
decode that $(w_{21},w_{22})=(\hat{k}_1,\hat{k}_2)$ were
transmitted; otherwise, an error is declared.

\subsection{Evaluation of Probability of Error}
We now derive upper bounds for the probabilities of the respective
error events, which may happen during the encoding and decoding
process. Due to the symmetry of the codebook generation and
encoding processing, the probability of error is not codeword
dependent. Without loss of generality, we assume that
$(w_1,w_{21},w_{22}) = (1, 1, 1)$ were encoded and transmitted. We
next define the following three types of events:
\begin{align*}
&E_{a,b}=(\bu(a), \bw(b), \bq)\in
A_{\epsilon}^{(n)},\\
&\dot{E}_{a,b} = (\bw(a),
\bu(b), \by_1, \bq)\in A_{\epsilon}^{(n)},\\
&\ddot{E}_{a,b} = ( \bu(a), \bv(b),\by_2, \bq)\in
A_{\epsilon}^{(n)}.
\end{align*}

Let $P_e(\text{enc2})$, $P_e(\text{dec1})$, and $P_e(\text{dec2})$
denote the probabilities of error at the encoder of sender 2, the
decoder of receiver 1, and the decoder of receiver 2,
respectively.

\noindent [\textbf {Evaluation of} $P_e(\text{enc2})$.] An error
is made if 1) the encoder at sender 2 can not find
$\bu(\hat{l}_1)$ in bin 1 such that $(\bu(\hat{l}_1), \bw(1),
\bq)\in A_{\epsilon}^{(n)}$, and/or 2) it can not find
$\bv(\hat{l}_2)$ in bin 1 such that $(\bv(\hat{l}_2), \bw(1),
\bq)\in A_{\epsilon}^{(n)}$. Then the probability of error at the
encoder of sender 2 is bounded as
\begin{align}
  P_e(\text{enc2}) &\leq Pr\left(\bigcap_{\bu(\hat{l}_1) \in
  \text{bin 1}}(\bu(\hat{l}_1), \bw(1), \bq)\notin
A_{\epsilon}^{(n)}\right) + Pr\left(\bigcap_{\bv(\hat{l}_2) \in
  \text{bin 1}}(\bv(\hat{l}_2), \bw(1), \bq)\notin
A_{\epsilon}^{(n)}\right) \notag \\
  &=\prod_{\bu(\hat{l}_1) \in \text{bin 1}}
  Pr(E^c_{\hat{l}_1,1}) + \prod_{\bv(\hat{l}_2) \in \text{bin 1}}
  Pr(E^c_{\hat{l}_2,1}) \nonumber\\
  &\leq
  (1-Pr(E_{\hat{l}_1,1}))^{2^{n(\tilde{R}_{21}-R_{21})}} + (1-Pr(E_{\hat{l}_2,1}))^{2^{n(\tilde{R}_{22}-R_{22})}}\nonumber \\
  &\stackrel{(a)}{\leq}
  (1-2^{-n(I(U;W|Q)+6\epsilon)})^{2^{n(\tilde{R}_{21}-R_{21})}}+
  (1-2^{-n(I(V;W|Q)+6\epsilon)})^{2^{n(\tilde{R}_{22}-R_{22})}},\notag
\end{align}
where (a) follows from the fact that we can obtain
$Pr(E_{\hat{l}_1,1}) \geq 2^{-n(I(U;W|Q)+6\epsilon)}$ and
$Pr(E_{\hat{l}_2,1}) \geq 2^{-n(I(V;W|Q)+6\epsilon)}$ by setting
$\bS_1'=\bU$, $\bS_2'=W$, and $\bS_3'=\bQ$, and $\bS_1'=\bV$,
$\bS_2'=W$, and $\bS_3'=\bQ$ in Lemma \ref{lemma_cover},
respectively. Following the same argument in the proof of Lemma
2.1.3 of \cite{Berger78:multiterm_source}, we conclude that
$P_e(\text{enc2})\rightarrow 0$ as $n\rightarrow +\infty$, if
\begin{align}
\tilde{R}_{21}& \geq R_{21}+I(U;W|Q), \label{P_e_enc2_new1}\\
\tilde{R}_{22}& \geq R_{22}+I(V;W|Q), \label{P_e_enc2_new2}
\end{align}
are satisfied. We further choose
\begin{align}
\tilde{R}_{21} = R_{21} + I(U;W|Q), \label{P_e_enc2_new1_eq}\\
\tilde{R}_{22} = R_{22} + I(V;W|Q). \label{P_e_enc2_new2_eq}
\end{align}
Note that such a choice still ensures that
$P_e(\text{enc2})\rightarrow 0$ as $n\rightarrow +\infty$.

\noindent [\textbf {Evaluation of} $P_e(\text{dec1})$] An error is
made if 1) $\dot{E}^c_{1,\hat{l}_1}$ happens, and/or 2) there
exists some $\hat{j}\neq 1$ such that
$\dot{E}_{\hat{j},\hat{\hat{l}}_1}$ happens. Note that
$\hat{\hat{l}}_1$ is not required to be equal to $\hat{l}_1$,
since it is unnecessary for receiver 1 to decode $\hat{l}_1$
correctly. The probability of error at receiver 1 can be upper
bounded as
\begin{align}
P_e(\text{dec1}) &\leq
Pr(\dot{E}^c_{1,\hat{l}_1}\bigcup\cup_{\hat{j}\neq 1} \dot{E}_{\hat{j},\hat{\hat{l}}_1})\notag \\
&\leq Pr(\dot{E}^c_{1,\hat{l}_1}) + \sum_{\hat{j}\neq 1}
Pr(\dot{E}_{\hat{j},\hat{\hat{l}}_1})\nonumber\\
&= Pr(\dot{E}^c_{1,\hat{l}_1}) + \sum_{\hat{j}\neq 1}
Pr(\dot{E}_{\hat{j},\hat{l}_1})+ \sum_{\hat{j}\neq 1,\hat{\hat{l}}_1
\neq \hat{l}_1}P(\dot{E}_{\hat{j},\hat{\hat{l}}_1})\nonumber\\
&\leq Pr(\dot{E}^c_{1,\hat{l}_1}) + 2^{nR_{1}}
Pr(\dot{E}_{2,\hat{l}_1}) + 2^{n(R_{1}+\tilde{R}_{21})}
Pr(\dot{E}_{2,\hat{\hat{l}}_1\neq\hat{l}_1}). \label{P_e_rec1_new}
\end{align}
Choosing $\bS_1'=\bW$, $\bS_2'=(\bY_1,\bU)$, and $\bS_3'=\bQ$ in
Lemma \ref{lemma_cover}, we have $Pr(\dot{E}_{2,\hat{l}_1}) \doteq
2^{-n(I(W;Y_1U|Q)\pm 6\epsilon)}$. Likewise, we have
$Pr(\dot{E}_{2,\hat{\hat{l}}_1\neq\hat{l}_1}) \doteq
2^{-n(I(WU;Y_1|Q)\pm 6\epsilon)}$. In addition, it follows from
AEP that $Pr(\dot{E}^c_{1,\hat{l}_1})\rightarrow 0$ as $n
\rightarrow +\infty$. Thus, we infer from \eqref{P_e_rec1_new}
that $P_e(\text{dec1}) \rightarrow 0$ as $n \rightarrow +\infty$,
if
\begin{align}
R_1 & \leq I(W;Y_1U|Q),\label{P_e_rec1_new_1}\\
R_1 + \tilde{R}_{21} & \leq I(WU;Y_1|Q), \label{P_e_rec1_new_3}
\end{align}
are satisfied.

\noindent [\textbf {Evaluation of} $P_e(\text{dec2})$] An error is
made if 1) $\ddot{E}^c_{\hat{l}_1,\hat{l}_2}$ happens, and/or 2)
there exists some $(\bar{\hat{l}}_1,\hat{\hat{l}}_2)$ in which
either $\bar{\hat{l}}_1$ or $\hat{\hat{l}}_2$ is not an index of
any codeword from the respective bin 1. The probability of the
second case is upper bounded by the probability of the event,
$\ddot{E}_{\bar{\hat{l}}_1,\hat{\hat{l}}_2}$ for some
$(\bar{\hat{l}}_1,\hat{\hat{l}}_2) \neq (\hat{l}_1,\hat{l}_2)$.
Thus, the probability of error at receiver 2 is bounded as
\begin{align}
P_e(\text{dec2}) &\leq
Pr(\ddot{E}^c_{\hat{l}_1,\hat{l}_2}\bigcup\cup_{(\bar{\hat{l}}_1,\hat{\hat{l}}_2)
\neq(\hat{l}_1,\hat{l}_2)}\ddot{E}_{\bar{\hat{l}}_1,\hat{\hat{l}}_2})\nonumber\\
&\leq
Pr(\ddot{E}^c_{\hat{l}_1,\hat{l}_2})+\sum_{(\bar{\hat{l}}_1,\hat{\hat{l}}_2)
\neq
(\hat{l}_1,\hat{l}_2)}P(\ddot{E}_{\bar{\hat{l}}_1,\hat{\hat{l}}_2})\nonumber\\
&= Pr(\ddot{E}^c_{\hat{l}_1,\hat{l}_2})+\sum_{\bar{\hat{l}}_1\neq
\hat{l}_1} Pr(\ddot{E}_{\bar{\hat{l}}_1,\hat{l}_2}) +
\sum_{\hat{\hat{l}}_2\neq \hat{l}_2}
Pr(\ddot{E}_{\hat{l}_1,\hat{\hat{l}}_2}) +
\sum_{(\bar{\hat{l}}_1\neq \hat{l}_1,\hat{\hat{l}}_2\neq \hat{l}_2)}
Pr(\ddot{E}_{\bar{\hat{l}}_1,\hat{\hat{l}}_2}) \nonumber\\
&\leq Pr(\ddot{E}^c_{\hat{l}_1,\hat{l}_2})+
2^{n\tilde{R}_{21}}Pr(\ddot{E}_{\bar{\hat{l}}_1 \neq
\hat{l}_1,\hat{l}_2}) +
2^{n\tilde{R}_{22}}Pr(\ddot{E}_{\hat{l}_1,\hat{\hat{l}}_2\neq
\hat{l}_2}) +
2^{n(\tilde{R}_{21}+\tilde{R}_{22})}Pr(\ddot{E}_{\bar{\hat{l}}_1\neq
\hat{l}_1,\hat{\hat{l}}_2\neq\hat{l}_2}). \label{P_e_rec2_new}
\end{align}
Applying Lemma \ref{lemma_cover} to evaluate
$Pr(\ddot{E}_{\bar{\hat{l}}_1 \neq \hat{l}_1,\hat{l}_2})$,
$Pr(\ddot{E}_{\hat{l}_1,\hat{\hat{l}}_2\neq \hat{l}_2})$ and
$Pr(\ddot{E}_{\bar{\hat{l}}_1\neq
\hat{l}_1,\hat{\hat{l}}_2\neq\hat{l}_2})$ in \eqref{P_e_rec2_new},
we conclude that $P_e(\text{dec2}) \rightarrow 0$ as $n
\rightarrow +\infty$ if the following inequalities,
\begin{align}
  \tilde{R}_{21} &\leq I(U;Y_2V|Q),\label{P_e_rec2_1_new}\\
  \tilde{R}_{22} &\leq I(V;Y_2U|Q),\\
  \tilde{R}_{21} + \tilde{R}_{22} & \leq I(UV;Y_2|Q),\label{P_e_rec2_3_new}
\end{align}
are satisfied.

According to \eqref{P_e_enc2_new1_eq}, \eqref{P_e_enc2_new2_eq}
and the fact that $R_2 = R_{21} + R_{22}$, we first substitute
$\tilde{R}_{21}$ and $\tilde{R}_{22}$ with $R_{21} + I(U;W|Q)$ and
$R_{22} + I(V;W|Q)$ in \eqref{P_e_rec1_new_1},
\eqref{P_e_rec1_new_3} and
\eqref{P_e_rec2_1_new}--\eqref{P_e_rec2_3_new}, and subsequently
substitute $R_{21}$ with $R_2 - R_{22}$ in the resulting
inequalities. After these two substitution steps, we have
\begin{align}
R_1 & \leq I(W;Y_1U|Q),\label{fourier_start}\\
R_1 + R_2-R_{22} & \leq I(WU;Y_1|Q)-I(U;W|Q),\\
R_2 - R_{22} &\leq I(U;Y_2V|Q) - I(U;W|Q),\\
R_{22} &\leq I(V;Y_2U|Q)-I(V;W|Q),\\
R_2 & \leq I(UV;Y_2|Q)-(I(U;W|Q)+I(V;W|Q)).\label{fourier_end}
\end{align}
Furthermore, applying Fourier-Motzkin elimination
\cite{kramer06_IZS:review_IFC} to remove $R_{22}$ from
\eqref{fourier_start}--\eqref{fourier_end}, we have
\begin{align}
R_1 &\leq I(W;Y_1U|Q), \label{region_ineq1_redun}\\
R_2 &\leq I(UV;Y_2|Q) - (I(U;W|Q)+I(V;W|Q)),\label{redundant_inq_not}\\
R_2 &\leq I(U;Y_2V|Q) - I(U;W|Q) + I
(V;Y_2U|Q) - I(V;W|Q),\label{redundant_inq}\\
R_1 + R_2 &\leq I(WU;Y_1|Q)-I(U;W|Q) + I(V;Y_2U|Q) - I (V;W|Q).
\label{region_ineq5_redun}
\end{align}
Since $I(U;Y_2V|Q)+I (V;Y_2U|Q)-I(UV;Y_2|Q)=I(U;V|Q)+I(U;V|Y_2Q)
\geq 0$, \eqref{redundant_inq_not} implies \eqref{redundant_inq}
and thus \eqref{redundant_inq} is redundant. To ensure that $R_1$,
$R_{21}$ and $R_{22}$ are non-negative, we enforce four additional
constraints \eqref{region_cons_1}--\eqref{region_cons_4}.
Therefore, the rate region $\mathcal{R}(p)$ is achievable for a
fixed joint probability distribution $p(\cdot)\in \mathcal{P}$,
and Theorem \ref{theorem_region_most} follows.
\end{proof}

\begin{Remark}
The proposed coding scheme exploits three coding methods to
achieve any rate pair in the rate region, $\mathcal{R}$. The first
method is {\it cooperation} that is realized by the superposition
relationship between $\bw$ and $\bx_2$ through
$p(x_2|\tilde{u}_2,\tilde{v}_2,w,q)$. The second is {\it
collaboration}, by which we mean that sender 2 separates its own
message into two parts, i.e., $w_2=(w_{21},w_{22})$, and encodes
 $w_{21}$ at a possibly low rate such that receiver 1 can decode
it. By doing so, the effective interference caused by the signals
carrying the sender 2's information may be reduced. The third is
{\it Gel'fand-Pinsker coding}, which we apply to encode both
messages, $w_{21}$ and $w_{22}$, from sender 2 by treating the
codeword $\bw$ as known interference. This perhaps allows receiver
2 to be able to decode the messages from sender 2 at the same rate
as if the interference caused by sender 1 was not present
\cite{Costa83:dirty_paper}.
\end{Remark}

\section{Relating with Existing Rate Regions}\label{section_region_general}

In this section, we will show that Theorem
\ref{theorem_region_most} includes the achievable rate regions in
\cite{jovicic06:cog_ICDMS,wuwei06_icdms}. To demonstrate it, we
compromise the advantages of the coding scheme developed in
Section \ref{section_region_general_MOST} to obtain the following
subregions of $\mathcal{R}$.

\subsection{A Subregion of $\mathcal{R}$}\label{section_region_sim}

Let $\mathcal{P}^*$ denote the set of all joint probability
distributions $p(\cdot)$ that factors in the form of
\begin{align}
p(q,w,x_1,u,v,\tilde{v},x_2,y_1,y_2)=&
  p(q)p(x_1,w|q)p(u|q)p(v,\tilde{v}|w,q) \notag \\
  &\cdot p(x_2|u,\tilde{v},w,q)p(y_1,y_2|x_1,x_2). \label{joint_pdf}
\end{align}
Note that the joint distribution \eqref{joint_pdf} differs from
\eqref{joint_pdf_most} in the way that conditioned on $Q$, $U$ is
now independent of any other auxiliary random variables, and
$\tilde{U}$ is not present.

Let $\mathcal{R}_{\sim}(p)$ denote the set of all non-negative
rate pairs $(R_1,R_2)$ such that
\begin{align}
  R_1 &\leq I(W;Y_1|UQ), \label{R_sim_ineq1}\\
  R_2 &\leq I(UV;Y_2|Q) - I(V;W|Q),\\
  R_1+R_2 &\leq
  I(WU;Y_1|Q)+I(V;Y_2|UQ)-I(V;W|Q);\label{R_sim_ineq3}\\
  0 &\leq I(V;Y_1|UQ) - I(V;W|Q),\label{R_sim_cons}
\end{align}
for a joint probability distribution $p(\cdot)\in \mathcal{P}^*$.
Furthermore, let
\[
\mathcal{R}_{\sim} = \bigcup_{p(\cdot)\in
\mathcal{P}^*}\mathcal{R}_{\sim}(p).
\]

\begin{Theorem}
  The rate region $\mathcal{R}_{\sim}$ is achievable for the discrete memoryless
  IC-DMS, i.e.,  $\mathcal{R}_{\sim}\subseteq \mathcal{R} \subseteq \mathcal{C}$.
  \label{theorem_region_general}
\end{Theorem}
\begin{proof}
The proof can be devised from the proof of Theorem
\ref{theorem_region_most} by customizing the original coding
scheme for the new joint distribution \eqref{joint_pdf}. We change
the encoding and decoding method for the message $w_{21}$
(corresponding to $U$), i.e., the Gel'fand-Pinsker coding used in
the proof of Theorem \ref{theorem_region_most} was replaced by the
conventional random coding. Specifically, we generate
$2^{nR_{21}}$ independent codewords $\bu(k_1)$, $k_1 \in
\{1,\ldots,2^{nR_{21}}\}$, according to
$\prod^n_{i=1}p(u_{i}|q_i)$. The encoding and decoding are then
adapted to the new codebook accordingly. Evaluating the
probability of error in the same way as was done in the proof of
Theorem \ref{theorem_region_most}, we obtain
\begin{align}
&\tilde{R}_{22}-R_{22} \geq I(V;W|Q);\label{eq_need}\\
&R_1  \leq I(W;Y_1|UQ),\label{fourier_sim_start}\\
&R_1 + R_{21}  \leq I(WU;Y_1|Q); \\
&R_{21} \leq I(U;Y_2|VQ),\\
&\tilde{R}_{22} \leq I(V;Y_2|UQ),\\
&R_{21} + \tilde{R}_{22}  \leq I(UV;Y_2|Q) \label{fourier_sim_end}.
\end{align}
Again, we choose $\tilde{R}_{22}-R_{22}=I(V;W|Q)$ in
\eqref{eq_need}, and then substitute $\tilde{R}_{22}$ with
$R_{22}+I(V;Y_2|UQ)$ as well as $R_{21}$ with $R_{2}-R_{22}$ in
the group of \eqref{fourier_sim_start}--\eqref{fourier_sim_end}.
By applying Fourier-Motzkin elimination on the resulting
inequalities to remove $R_{22}$, and adding the constraints that
ensure the respective rates $R_1$, $R_{21}$ and $R_{22}$ to be
non-negative, we obtain \eqref{R_sim_ineq1}--\eqref{R_sim_cons}.
Therefore, the region $R_{\sim}(p)$ is achievable for a given
$p(\cdot) \in \mathcal{P}^*$, and the theorem follows.
\end{proof}

Note that simultaneous decoding (simultaneous joint typicality) is
applied at both decoders. The advantage of simultaneous decoding
over successive decoding is well demonstrated on the IC by Han and
Kobayashi in \cite{Han81:IFC}. We next modify the coding scheme by
applying successive decoding instead of simultaneous decoding at
both decoders to derive a subregion of $\mathcal{R}_{\sim}$.

\subsection{A Subregion of $\mathcal{R}_{\text{sim}}$}\label{section_region_compromised}

Let $\mathcal{R}_{\suc}(p)$ denote the set of all achievable rate
pairs $(R_1, R_2)$ such that
\begin{align}
  &R_1  \leq I(W;Y_1|UQ),\label{region_comp_ineq1_reduced}\\
  &R_2  \leq \min\{I(U;Y_1|Q), I(U;Y_2|Q)\} +
  I(V;Y_2|UQ)-I(V;W|Q); \label{region_comp_ineq2_reduced}\\
  &0 \leq I(V;Y_1|UQ) - I(V;W|Q), \label{region_comp_cons_reduced}
\end{align}
for a fixed joint probability distribution $p(\cdot)\in
\mathcal{P}^*$. Define
  \begin{align*}
    \mathcal{R}_{\suc} = \bigcup_{p(\cdot)\in
  \mathcal{P}^*}\mathcal{R}_{\suc}(p).
  \end{align*}

\begin{Theorem}
The rate region $\mathcal{R}_{\suc}$ is achievable for the
discrete memoryless IC-DMS, i.e., $\mathcal{R}_{\suc} \subseteq
\mathcal{R}_{\sim} \subseteq \mathcal{R} \subseteq \mathcal{C}$.
\label{theorem_region_comp}
\end{Theorem}
\begin{proof}
The codebook generation, encoding and transmission remain the same
as those used to prove Theorem \ref{theorem_region_general},
whereas the decoding processes at both decoders are altered. Both
decoders decode $w_{21}$ first, and then decoder 1 decodes $w_1$
and decoder 2 decodes $w_{22}$ respectively. Then the following
can easily be obtained
\begin{align}
&\tilde{R}_{22}-R_{22} \geq I(V;W|Q),\label{inq_suc_start}\\
&R_{21} \leq I(U;Y_1|Q),\label{inq_suc_special}\\
&R_1 \leq I(W;Y_1|UQ), \\
&R_{21}  \leq I(U;Y_2|Q),\\
&\tilde{R}_{22} \leq I(V;Y_2|UQ)\label{inq suc_end}.
\end{align}
From \eqref{inq_suc_start}--\eqref{inq suc_end}, it is
straightforward to obtain
\eqref{region_comp_ineq1_reduced}--\eqref{region_comp_cons_reduced}.
Therefore, the region $\mathcal{R}_{\suc}(p)$ is achievable, and
the theorem follows immediately.
\end{proof}

\begin{Remark}
Note that \eqref{inq_suc_special} is only necessary when the
successive decoding is applied. This is because every decoding
step in a successive decoding scheme is expected to have a
vanishing probability of error.
\end{Remark}

In what follows, we further specialize the subregion
$\mathcal{R}_{\suc}$ to obtain two more achievable rate regions
$\mathcal{R}_{\sptfirst}$ and $\mathcal{R}_{\spttwo}$. Let
$\mathcal{P}^*_1$ denote the set of all joint probability density
distributions $p(\cdot)$ that factor in the form of
\begin{align}
    p(q,w,x_1,v,\tilde{v},x_2,y_1,y_2) =
  p(q)p(x_1,w|q)p(v,\tilde{v}|w,q)p(x_2|\tilde{v},w,q)p(y_1,y_2|x_1,x_2). \label{joint_pdf_sp1}
\end{align}
Let $\mathcal{R}_{\sptfirst}(p)$ denote the set of all
non-negative rate pairs $(R_1, R_2)$ such that
\begin{align}
  R_1 & \leq I(W;Y_1|Q),\label{region_sp1_ineq1}\\
  R_2 & \leq I(V;Y_2|Q)-I(V;W|Q), \label{region_sp1_ineq2}
\end{align}
for a fixed joint distribution $p(\cdot) \in \mathcal{P}^*_1$.
Define
\[
\mathcal{R}_{\sptfirst} = \bigcup_{p(\cdot) \in \mathcal{P}^*_1}
  \mathcal{R}_{\sptfirst}(p).
\]
\begin{Corollary}
The region $\mathcal{R}_{\sptfirst}$ is an achievable rate region
for the discrete memoryless IC-DMS, i.e., $\mathcal{R}_{\sptfirst}
\subseteq \mathcal{R}_{\suc} \subseteq \mathcal{R}_{\sim}
  \subseteq \mathcal{R} \subseteq
  \mathcal{C}$. \label{corollary_sp1}
\end{Corollary}
\begin{proof}
  Fixing the auxiliary random variable $U$
  as a constant, we reduce \eqref{region_comp_ineq1_reduced}
  and \eqref{region_comp_ineq2_reduced} to
  \eqref{region_sp1_ineq1} and \eqref{region_sp1_ineq2}, and the
  corollary follows immediately.
\end{proof}

\begin{Remark}
We note that the region $\mathcal{R}_{\sptfirst}$ is similar to
the region $\mathcal{R}_{\inn}$ reported in \cite[Theorem
3.1]{wuwei06_icdms}. It seems that the region $\mathcal{R}_{\inn}$
is more general than the region $\mathcal{R}_{\sptfirst}$ in the
sense that fixing the auxiliary random variable $U$ in
$\mathcal{R}_{\inn}$ as a constant, one can obtain a region which
is the same as $\mathcal{R}_{\sptfirst}$. Nevertheless, after
examining the coding scheme used in \cite[Theorem
3.1]{wuwei06_icdms}, one can find that there exists a one-one
correspondence between codewords $\bu(w_2)$ and $\bx_2(w_2)$, and
both codewords are jointly generated and decoded, i.e., $p(u,x_2)$
is used to generate two-letter codewords. Thus, one can introduce
one auxiliary random variable $W$ such that there exists a one-one
mapping between $\mathcal{W}$ and $\mathcal{U} \times
\mathcal{X}_2$, i.e., $f:\mathcal{U} \times \mathcal{X}_2
\leftrightarrow \mathcal{W} $, and thus $W$ has the probability
mass distribution $p(w) = p(f^{-1}(w))=p(u,x_2)$. Replacing all
$(X_2^n(w_2),U^n(w_2))$ by $W^n(w_2)$ in the proof of
\cite[Theorem 3.1]{wuwei06_icdms} will yield the same rate region.
Equivalently speaking, for any input distribution achieving a rate
region characterized by \cite[Theorem 3.1]{wuwei06_icdms}, one can
find a corresponding joint distribution in the form of
\eqref{joint_pdf_sp1} such that Corollary \ref{corollary_sp1}
yields exactly the same rate region. Therefore, two rate regions
$\mathcal{R}_{\sptfirst}$ and $\mathcal{R}_{\inn}$ are identical.
\end{Remark}

Let $\mathcal{P}^*_2$ denote the set of all joint probability
distributions $p(\cdot)$ that factor in the form of
\begin{align}
    p(q,w,x_1,u,x_2,y_1,y_2) =
  p(q)p(x_1,w|q)p(u|q)p(x_2|u,w,q)p(y_1,y_2|x_1,x_2). \label{joint_pdf_sp2}
\end{align}
Let $\mathcal{R}_{\spttwo}(p)$ denote the set of all non-negative
rate pairs $(R_1, R_2)$ such that
\begin{align}
  R_1 & \leq I(W;Y_1|UQ),\label{region_sp2_ineq1}\\
  R_2 & \leq \min\{I(U;Y_1|Q), I(U;Y_2|Q)\}, \label{region_sp2_ineq2}
\end{align}
for a fixed joint distribution $p(\cdot) \in \mathcal{P}^*_2$.
Define
\[
\mathcal{R}_{\spttwo} = \bigcup_{p(\cdot) \in \mathcal{P}^*_2}
  \mathcal{R}_{\spttwo}(p).
\]
\begin{Corollary}
  The region $\mathcal{R}_{\spttwo}$ is an achievable rate region for the
  discrete memoryless IC-DMS, i.e., $\mathcal{R}_{\spttwo}  \subseteq \mathcal{R}_{\suc} \subseteq \mathcal{R}_{\sim}
  \subseteq \mathcal{R} \subseteq
  \mathcal{C}$. \label{corollary_sp2}
\end{Corollary}
\begin{proof}
  The proof can be devised from the proof of Theorem
  \ref{theorem_region_comp} easily by fixing $V$ as a constant.
\end{proof}

\section{The Gaussian IC-DMS}\label{section_region_Gaussian}
In the preceding sections, we have derived several achievable rate
regions for the discrete memoryless IC-DMS. We now extend these
results to obtain corresponding achievable rate regions for the
{\it Gaussian} IC-DMS (GIC-DMS).

\subsection{The Channel Model of the GIC-DMS}
In general, with no loss of information-theoretic optimality, the
GIC-DMS can be converted to the GIC-DMS in the standard form
through invertible transformations \cite{jovicic06:cog_ICDMS,
kramer06_IZS:review_IFC, Carleial78:IFC}. We thus only consider
the GIC-DMS in the standard form, which is represented as follows
\begin{align*}
Y_1 &= X_1 + \sqrt{c_{21}}X_2 + Z_1, \notag \\
Y_2 &= X_2 +\sqrt{c_{12}}X_1 + Z_2,
\end{align*}
where $Z_i$, $i=1,2$, is the additive white Gaussian noise with
zero mean and unit variance, and $\sqrt{c_{21}}$ and
$\sqrt{c_{12}}$ are the {\it normalized} link gains in the GIC-DMS
depicted in Fig. \ref{fig_ICDMS_Gaussian}.
\begin{figure}[h]
\centering \psfrag{X1}{$\bx_1(w_1)$} \psfrag{X2}{$\bx_2(w_2,w_1)$}
\psfrag{Y1}{$\by_1$} \psfrag{Y2}{$\by_2$}\psfrag{Z1}{$\bz_1$}
\psfrag{Z2}{$\bz_2$}\psfrag{c12}{$\sqrt{c_{12}}$}
\psfrag{c21}{$\sqrt{c_{21}}$}
\includegraphics{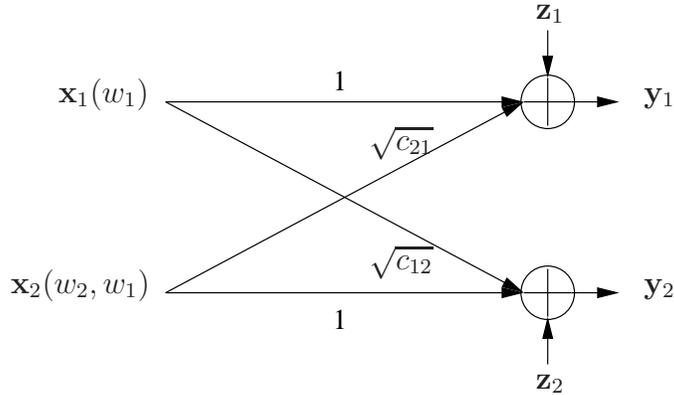}
\caption{A Gaussian interference channel with degraded message
sets.} \label{fig_ICDMS_Gaussian}
\end{figure}
Moreover, the transmitted codeword $\bx_i=(x_{i1},\ldots,x_{in})$,
$ i=1,2$, is subject to an average power constraint given by
\begin{align*}
  \frac{1}{n} \sum_{t=1}^n \|{x_{it}}\|^2 \leq P_i.
\end{align*}
Since it has been shown in the maximum-entropy theorem in
\cite{cover_IT_book} that Gaussian input signals are optimal for
Gaussian channels, we will consider Gaussian codewords $X^n_i$,
$i=1,2$.

\subsection{Achievable Rate Regions for the GIC-DMS}

\subsubsection{Gaussian Extension of $\mathcal{R}$}

We first extend $\mathcal{R}$ to its Gaussian counterpart denoted
by $\mathcal{G}$. To obtain the rate region $\mathcal{G}$, we map
the random variables involved in the joint distribution
\eqref{joint_pdf_most} to the corresponding Gaussian random
variables with the following customary constraints:
\begin{enumerate}[P1)]
\item $W$, distributed according to $\mathcal{N}(0,1)$,
\item $X_1=\sqrt{P_1}W$,
\item $\tilde{U}$, distributed according to $\mathcal{N}(0,\alpha \beta
P_2)$,
\item $\tilde{V}$, distributed according to $\mathcal{N}(0,\alpha \bar{\beta}
P_2)$,
\item $U$ = $\tilde{U} + \lambda_1 W$,
\item $V$ = $\tilde{V} + \lambda_2 W$,
\item $X_2 = \tilde{U} + \tilde{V} +
\sqrt{\bar{\alpha}P_2}W$,
\end{enumerate}
where $\alpha, \beta \in [0,1]$, $\alpha + \bar{\alpha}=1$, $\beta
+ \bar{\beta}=1$, $\lambda_1, \lambda_2 \in [0,+\infty)$, and $W$,
$\tilde{U}$ and $\tilde{V}$ are mutually independent. The
input-output relationship of the GIC-DMS can be described by
\begin{align}
  Y_1 &= \left(\sqrt{P_1}+ \sqrt{c_{21}\bar{\alpha}P_2}\right)W +
  \sqrt{c_{21}}\tilde{U} + \sqrt{c_{21}} \tilde{V} + Z_1,\\
 Y_2 &= \tilde{U} + \tilde{V} +
 \left(\sqrt{\bar{\alpha}P_2}+\sqrt{c_{12}P_1}\right)W + Z_2.
\end{align}
To simplify the derivations, we fix the time-sharing random
variable $Q$ as a constant. The issue of how this time-sharing
random variable affects the achievable rate region is well
addressed in \cite{sason04:ifc_gaussian_inner}. In the Gaussian
case, the respective mutual information terms in
\eqref{region_ineq1} -- \eqref{region_cons_4} need be evaluated
with respect to the mappings defined by P1--P7. Since the
computation procedure to obtain $\mathcal{G}$ and the resulting
description of $\mathcal{G}$ are fairly lengthy, we relegate them
(Theorem 5) to the Appendix.

\subsubsection{Gaussian Extension of $\mathcal{R}_{\suc}$}
For illustration and comparison purpose, we next show how to
obtain the Gaussian counterpart of $\mathcal{R}_{\suc}$ in
details. Following the first step in the previous derivation, we
also map the random variables involved in \eqref{joint_pdf} to the
Gaussian ones with the following constraints:
\begin{enumerate}[M1)]
\item $W$, distributed according to $\mathcal{N}(0,1)$,
\item $X_1=\sqrt{P_1}W$,
\item $U$, distributed according to $\mathcal{N}(0,\alpha \beta
P_2)$,
\item $\tilde{V}$, distributed according to $\mathcal{N}(0,\alpha \bar{\beta}
P_2)$,
\item $V$ = $\tilde{V} + \lambda W$,
\item $X_2 = U + \tilde{V} +
\sqrt{\bar{\alpha}P_2}W$,
\end{enumerate}
where $\alpha, \beta \in [0,1]$, $\alpha + \bar{\alpha}=1$, $\beta
+ \bar{\beta}=1$, $\lambda \in [0,+\infty)$, and $W$, $U$ and
$\tilde{V}$ are mutually independent. Using the mappings defined
by M1--M6, we express the input-output relationship for the
GIC-DMS as:
\begin{align}
  Y_1 &= \left(\sqrt{P_1}+ \sqrt{c_{21}\bar{\alpha}P_2}\right)W +
  \sqrt{c_{21}}U + \sqrt{c_{21}} \tilde{V} + Z_1,\\
 Y_2 &= U + \tilde{V} +
 \left(\sqrt{\bar{\alpha}P_2}+\sqrt{c_{12}P_1}\right)W + Z_2.
\end{align}

Let $\mathcal{G}_{\suc}(\alpha, \beta)$ denote the set of all the
non-negative rate pairs $(R_1,R_2)$ such that
\begin{align}
  R_1 \leq &\frac{1}{2} \log_2 \left(1+ \frac{\left(\sqrt{P_1}+ \sqrt{c_{21}\bar{\alpha}P_2}\right)^2}{c_{21}\alpha\bar{\beta}P_2+1}\right),\label{region_suc_gaussian_ineq1}
\\
  R_2 \leq &\frac{1}{2} \log_2 (1+ \alpha\bar{\beta}P_2)
  + \min\Bigg\{
  \frac{1}{2} \log_2 \left(1+ \frac{c_{21}\alpha\beta P_2}{\left(\sqrt{P_1}+ \sqrt{c_{21}\bar{\alpha}P_2}\right)^2
  +c_{21}\alpha\bar{\beta}P_2+1}\right), \notag \\
   &\qquad \qquad \qquad \quad \qquad \qquad \frac{1}{2} \log_2\left(1+\frac{\alpha\beta P_2}{\alpha\bar{\beta}P_2+\left(\sqrt{\bar{\alpha}P_2}+\sqrt{c_{12}P_1}\right)^2+1}\right)
   \Bigg\}. \label{region_suc_gaussian_ineq2}
\end{align}
Define
\[
\mathcal{G}_{\suc} = \bigcup_{\alpha,\beta \in
[0,1]}\mathcal{G}_{\suc}(\alpha, \beta).
\]
\begin{Theorem}
The region $\mathcal{G}_{\suc}$ is an achievable rate region for
the GIC-DMS in the standard form.
\label{theorem_region_comp_Gaussian}
\end{Theorem}
\begin{proof}
It suffices to prove that $\mathcal{G}_{\suc}(\alpha, \beta)$ is
achievable for any given $\alpha, \beta \in [0,1]$. Since
$\mathcal{G}_{\suc}$ is extended from $\mathcal{R}_{\suc}$, we
need compute the mutual information terms in
\eqref{region_comp_ineq1_reduced} and
\eqref{region_comp_ineq2_reduced}. The righthand side of
\eqref{region_suc_gaussian_ineq1} can be readily obtained through
a straightforward computation of $I(W;Y_1|UQ)$ in
\eqref{region_comp_ineq1_reduced}. Recall that $Q$ is a constant.
It is also fairly straightforward to obtain the two terms within
the minimum operator in \eqref{region_suc_gaussian_ineq2} through
computing $I(U;Y_1|Q)$ and $I(U;Y_2|Q)$ in
\eqref{region_comp_ineq2_reduced}. We next evaluate the only
remaining term $I(V;Y_2|UQ)-I(V;W|Q)$ for a constant $Q$. Defining
$\tilde{Y}_2 =  \tilde{V} +
 \left(\sqrt{\bar{\alpha}P_2}+\sqrt{c_{12}P_1}\right)W + Z_2$, we
 have
  \begin{align}
   I(V;Y_2|U)-I(V;W) &= h(Y_2|U) - h(Y_2|UV) -
   I(V;W) \notag \\
   &= h(\tilde{Y}_2) - h(\tilde{Y}_2|V) -
   I(V;W) \nonumber \\
    &= h(\tilde{Y}_2) + h(V) - h(\tilde{Y}_2V) -
   I(V;W). \label{term_evaluation_G_suc_2}
  \end{align}
With $V = \tilde{V}+ \lambda W$, we evaluate
\eqref{term_evaluation_G_suc_2} as
\begin{align}
I(V&;Y_2|U)-I(V;W) \nonumber\\
=&\frac{1}{2} \log_2 \left(2\pi e \left(\alpha \bar{\beta} P_2 +
\left(\sqrt{\bar{\alpha}P_2}+\sqrt{c_{12}P_1}\right)^2+1\right)\right)
+ \frac{1}{2} \log_2 (2\pi e (\alpha \bar{\beta} P_2 + \lambda^2
))\nonumber \\
&- \frac{1}{2} \log_2 \Bigg((2\pi e)^2 \Bigg[\left(\alpha
\bar{\beta} P_2 + \left(\sqrt{\bar{\alpha}P_2}+
\sqrt{c_{12}P_1}\right)^2+1\right)(\alpha \bar{\beta} P_2 +
\lambda^2
P_1)\nonumber\\
&~~~~~~~~~~~~~~-\left(\alpha \bar{\beta} P_2+ \lambda
\left(\sqrt{\bar{\alpha}P_2}+\sqrt{c_{12}P_1}\right)
\right)^2\Bigg]\Bigg)- \frac{1}{2} \log_2 \left(1+
\frac{\lambda^2}{\alpha \bar{\beta} P_2}\right).
\end{align}
It is easy to find that when
\begin{align}
\lambda =\frac{\alpha \bar{\beta}P_2
\left(\sqrt{\bar{\alpha}P_2}+\sqrt{c_{12}P_1}\right)}{\alpha
\bar{\beta}P_2+1},
\end{align}
the term $I(V;Y|U)-I(V;W)$ is maximized, and the maximum value is
\begin{align}
\max[I(V;Y_2|U)-I(V;W)] = \frac{1}{2} \log_2 (1+
\alpha\bar{\beta}P_2).
\end{align}
This is in parallel with the result in \cite{Costa83:dirty_paper}.

Therefore, the rate region $\mathcal{G}_{\suc}(\alpha, \beta)$ is
achievable for any pair $\alpha, \beta \in [0,1]$, and the theorem
follows.
\end{proof}

In the following, we obtain two corollaries by setting $\beta = 0$
and $\beta = 1$ in Theorem \ref{theorem_region_comp_Gaussian},
respectively.

\begin{Corollary}
  The rate region $\mathcal{G}_{\sptfirst}$ is an achievable rate region
  for the GIC-DMS in the standard form with $\mathcal{G}_{\sptfirst}: = \bigcup_{\alpha \in
  [0,1]}\mathcal{G}_{\suc}(\alpha, 0)$, i.e., $\mathcal{G}_{\sptfirst}$ is
  the union of the sets of non-negative rate pairs $(R_1,R_2)$ satisfying
  \begin{align*}
  R_1 &\leq \frac{1}{2} \log_2 \left(1+ \frac{\left(\sqrt{P_1}+
  \sqrt{c_{21}\bar{\alpha}P_2}\right)^2}{c_{21}\alpha
  P_2+1}\right),
  \\
  R_2 &\leq \frac{1}{2} \log_2 (1+ \alpha P_2),
  \end{align*}
  over all $\alpha \in
  [0,1]$. \label{corollary_sp1_Gaussian}
\end{Corollary}

\begin{Corollary}
  The rate region $\mathcal{G}_{\spttwo}$ is an achievable rate region
  for the GIC-DMS in the standard form with $\mathcal{G}_{\spttwo}: = \bigcup_{\alpha \in
  [0,1]}\mathcal{G}_{\suc}(\alpha, 1)$, i.e., $\mathcal{G}_{\spttwo}$ is
  the union of the sets of non-negative rate pairs $(R_1,R_2)$ satisfying
  \begin{align*}
  R_1 \leq &\frac{1}{2} \log_2 \left(1+ \left(\sqrt{P_1}+
  \sqrt{c_{21}\bar{\alpha}P_2}\right)^2\right),\\
  R_2 \leq &\min\left\{\frac{1}{2} \log_2 \left(1+ \frac{c_{21}\alpha P_2}{\left(\sqrt{P_1}+
  \sqrt{c_{21}\bar{\alpha}P_2}\right)^2+1}\right),
   \frac{1}{2} \log_2\left(1+\frac{\alpha P_2}{\left(\sqrt{\bar{\alpha}P_2}+\sqrt{c_{12}P_1}\right)^2+1}\right)
   \right\},
  \end{align*}
  over all $\alpha \in
  [0,1]$. \label{corollary_sp2_Gaussian}
\end{Corollary}

\begin{Remark}
Corollaries \ref{corollary_sp1_Gaussian} and
\ref{corollary_sp2_Gaussian} correspond the Gaussian extensions of
Corollaries \ref{corollary_sp1} and \ref{corollary_sp2}
respectively. Particularly, the rate region depicted by Corollary
\ref{corollary_sp1_Gaussian} is the same as the rate regions given
in \cite[Theorem 4.1]{jovicic06:cog_ICDMS} and \cite[Theorem
3.5]{wuwei06_icdms}. It has been proven in both
\cite{jovicic06:cog_ICDMS} and \cite{wuwei06_icdms} that the rate
region $\mathcal{G}_{\sptfirst}$ is indeed the capacity region for
the GIC-DMS in the low-interference-gain regime, i.e., $c_{21}
\leq 1$.

In addition,  the set of achievable rate pairs given in
\cite[Lemma 4.2]{jovicic06:cog_ICDMS} is contained in the region
$\mathcal{G}_{\spttwo}$ as a subset.
\end{Remark}

\begin{figure}[t]
\includegraphics{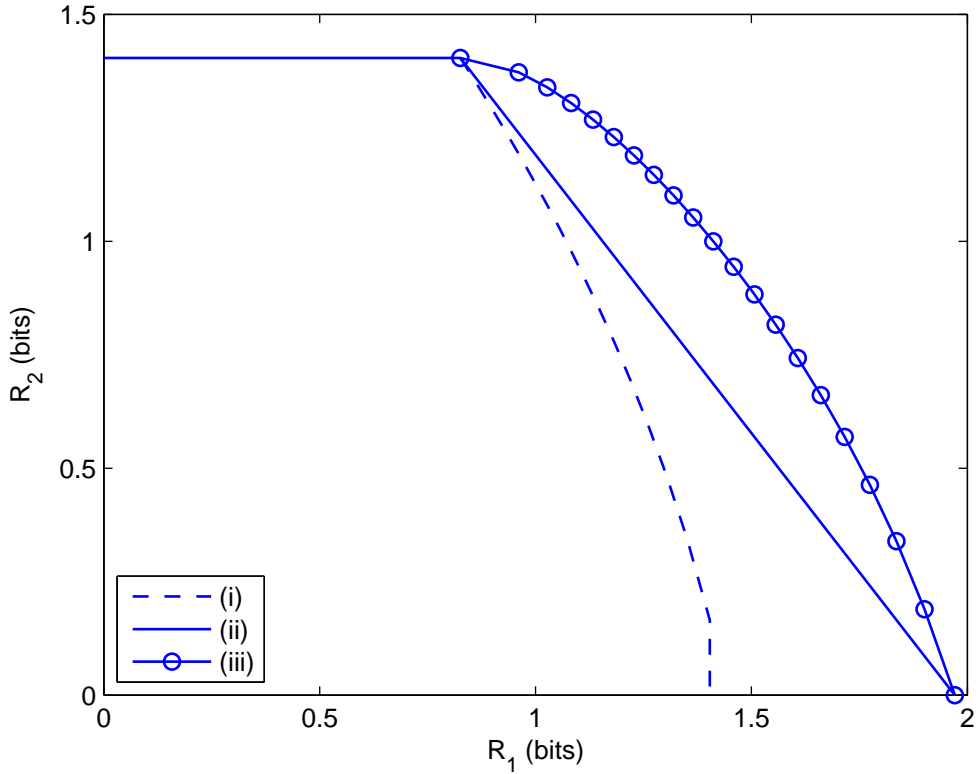}
\caption{$P_1 = P_2= 6$, $c_{21} = 0.3$, $c_{12} = 0$. (i) gives
the rate region in Theorem 1 of \cite{Tarokh06:ic_dms_cog}; (ii)
gives the rate region in Corollary 2 of
\cite{Tarokh06:ic_dms_cog}; (iii) gives the rate region in
Corollary 3 (equivalently, Theorem 4.1 of
\cite{jovicic06:cog_ICDMS} and Theorem 3.5 of
\cite{wuwei06_icdms}).} \label{fig_compare_tarokh}
\end{figure}
\subsection{Numerical Examples}
We next provide several numerical examples to illustrate
improvements of our achievable rate regions over the previously
known results in
\cite{Tarokh06:ic_dms_cog,jovicic06:cog_ICDMS,wuwei06_icdms}.
Denote the achievable rate regions obtained in \cite[Theorem
1]{Tarokh06:ic_dms_cog} and \cite[Corollary2]{Tarokh06:ic_dms_cog}
by $\mathcal{G}_{\text{dmt1}}$ and $\mathcal{G}_{\text{dmt2}}$,
respectively.

\subsubsection{Comparing with Rate Regions in
\cite{Tarokh06:ic_dms_cog}} Fig. \ref{fig_compare_tarokh} compares
the rate regions $\mathcal{G}_{\text{dmt1}}$,
$\mathcal{G}_{\text{dmt2}}$, and $\mathcal{G}_{\sptfirst}$ for an
extreme case in which receiver 2 does not experience any
interference from sender 1, i.e., $c_{12} = 0$.  As can be seen
from Fig. \ref{fig_compare_tarokh}, the rate region
$\mathcal{G}_{\sptfirst}$ strictly includes
$\mathcal{G}_{\text{dmt1}}$, as well as
$\mathcal{G}_{\text{dmt2}}$ obtained through time-sharing between
$\mathcal{G}_{\text{dmt1}}$ and a fully-cooperative rate point.
The coding scheme used to establish $\mathcal{G}_{\text{dmt1}}$
incurs certain rate loss due to the fact that sender 2 does not
use its power to help the sender 1's transmissions even though it
has complete and non-causal knowledge about the message being
transmitted by sender 1. In contrast, our proposed coding scheme
allows sender 2 to use superposition coding to help sender 1, and
thus yields an improved rate region.

In Fig. \ref{fig_compare_tarokh2}, we consider another case in
which the transmit power of sender 1 is set to zero and $c_{21}
\leq 1$. From the figure, we observe that the rate region
$\mathcal{G}_{\text{dmt2}}$ is strictly smaller than
$\mathcal{G}_{\sptfirst}$. Note that in this case, the GIC-DMS
becomes a Gaussian degraded broadcast channel. According to
\cite{cover_IT_book}, the optimal coding scheme for this case is:
sender 2 uses a portion of its power to transmit the codeword
conveying $w_1$, and uses the remaining power to transmit the
codeword conveying $w_2$, which is encoded by using the
dirty-paper coding \cite{yuwei01:bc_dirty_paper}. It is easy to
verify that this scheme is a special case of the coding scheme
developed in Theorem \ref{theorem_region_most}.

\begin{figure}[h]
\includegraphics{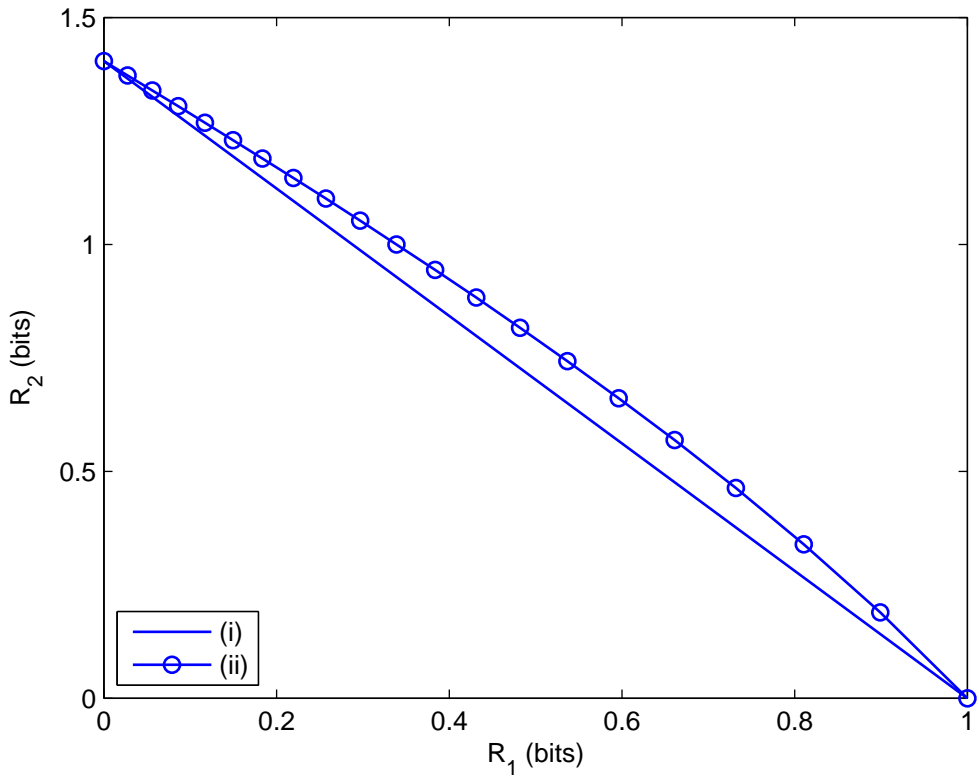}
\caption{$P_1 =0$, $ P_2= 6$, $c_{21} = 0.5$. (i) gives the rate
region in Corollary 2 of \cite{Tarokh06:ic_dms_cog}; (ii) gives the
rate region in Corollary 3 (equivalently, Theorem 4.1 of
\cite{jovicic06:cog_ICDMS} and Theorem 3.5 of
\cite{wuwei06_icdms}).} \label{fig_compare_tarokh2}
\end{figure}

\subsubsection{Comparing with Rate Regions in
\cite{jovicic06:cog_ICDMS,wuwei06_icdms}}

\begin{figure}[h]
\includegraphics{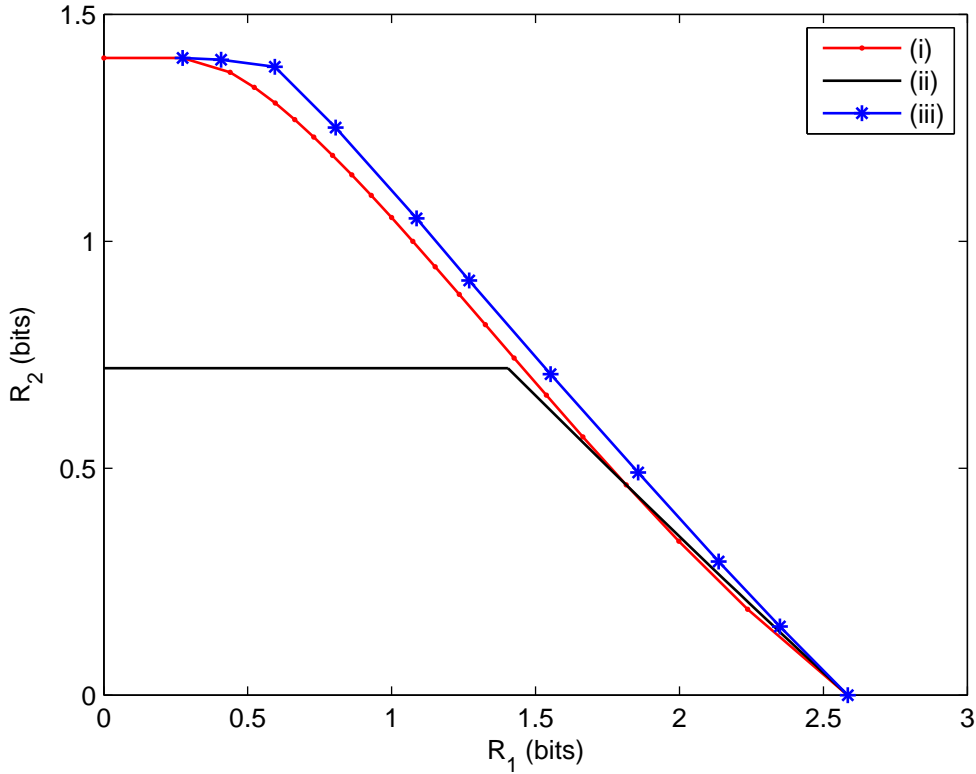}
\caption{$P_1 = P_2= 6$, $c_{21} = 2$, $c_{12} = 0.3$. (i) gives
the rate region in Corollary 3 (equivalently, Theorem 4.1 of
\cite{jovicic06:cog_ICDMS} and Theorem 3.5 of
\cite{wuwei06_icdms}); (ii) gives the achievable rate region in
Corollary 4; (iii) gives the achievable rate region in Theorem 5.}
\label{fig_compare_big_region_2}
\end{figure}

As mentioned earlier, the rate region $\mathcal{G}_{\sptfirst}$, a
subregion of $\mathcal{G}$, is the same as the one given in
\cite[Theorem 4.1]{jovicic06:cog_ICDMS} and the one given in
\cite[Theorem 3.5]{wuwei06_icdms}, which is indeed the capacity
region for GIC-DMS in the low-interference-gain regime. In Figs.
\ref{fig_compare_big_region_2} and \ref{fig_compare_big_region_6},
we compare $\mathcal{G}$ with $\mathcal{G}_{\sptfirst}$ and
$\mathcal{G}_{\spttwo}$ in the high-interference-gain regime,
i.e., $c_{21}> 1$. As can be seen from the figures, the rate
region $\mathcal{G}$ strictly includes both
$\mathcal{G}_{\sptfirst}$ and $\mathcal{G}_{\spttwo}$ in this
case. Comparing Fig. \ref{fig_compare_big_region_2} with Fig.
\ref{fig_compare_big_region_6}, we observe that the improvement of
the rate region $\mathcal{G}$ over $\mathcal{G}_{\sptfirst}$
becomes more pronounced as the link gain $c_{21}$ increases. The
improvement is mainly because our coding scheme allows receiver 1
to decode partial information from sender 2, and thus reduces the
effective interference experienced by receiver 1. In addition, it
can be seen from the figures that in the high-interference-gain
regime, $\mathcal{G}_{\sptfirst}$ is not convex and thus is only
suboptimal.

\begin{figure}[h]
\includegraphics{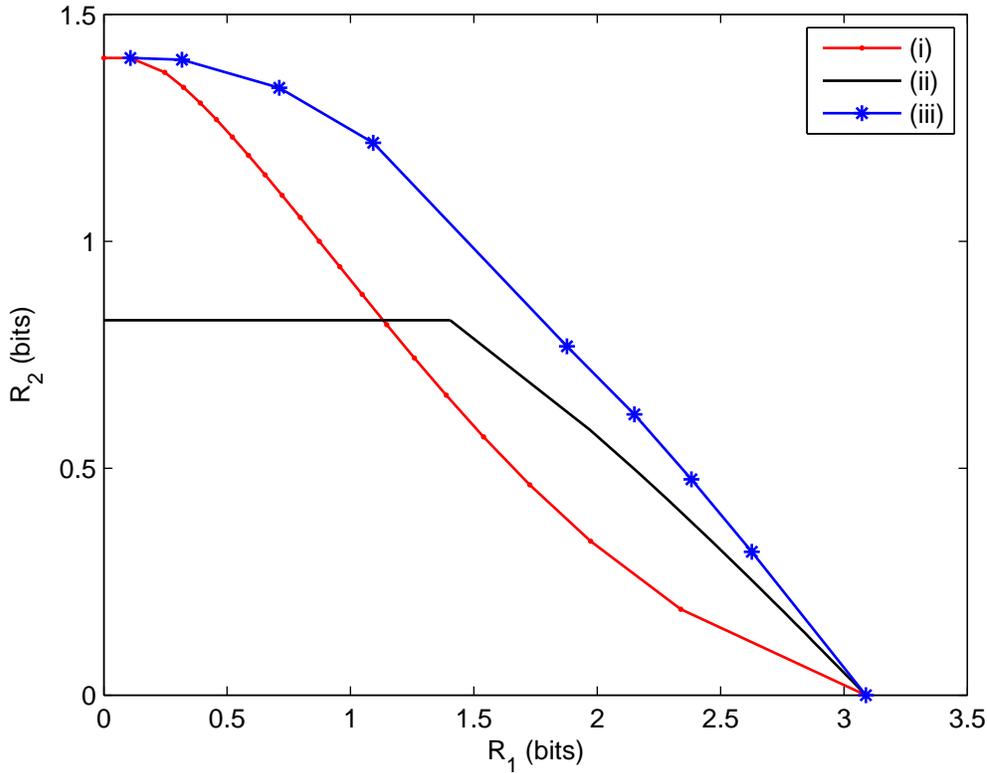}
\caption{$P_1 = P_2= 6$, $c_{21} = 6$, $c_{12} = 0.3$. (i) gives the
rate region in Corollary 3 (equivalently, Theorem 4.1 of
\cite{jovicic06:cog_ICDMS} and Theorem 3.5 of \cite{wuwei06_icdms});
(ii) gives the achievable rate region in Corollary 4; (iii) gives
the achievable rate region in Theorem 5.}
\label{fig_compare_big_region_6}
\end{figure}

\section{Conclusions}\label{section_conclusion}
In this paper, we have investigated the IC-DMS from an information
theoretic perspective. We have developed a coding scheme that
combines the advantages of cooperative coding, collaborative
coding and Gel'fand-Pinsker coding. With the coding scheme, we
have derived a new achievable rate region for such a channel,
which not only includes existing results as special cases, but
also exceeds them in the high-interference-gain regime. However,
we are not able to establish a converse for the derived achievable
rate region, because the achievable result is closely related to
the achievable results for the interference channel and the
broadcast channel, for which there is no converse available in
general.

\section*{Appendix\\ An Achievable Rate Region for the GIC-DMS}
In this appendix, we show how to extend $\mathcal{R}$, the
achievable rate region for the discrete memoryless IC-DMS, to its
Gaussian counterpart, $\mathcal{G}$. Note that the mappings M1--M6
of the auxiliary random variables are described in Section
\ref{section_region_Gaussian}. We first compute the following two
covariance matrices:
\begin{align*}
  {\mathbf{\Sigma}}_{WUY_1} &=\left( \begin{array}{ccc}
\mu_{11} & \mu_{12} & \mu_{13} \\
\mu_{21} & \mu_{22} & \mu_{23}\\
\mu_{31} & \mu_{32} & \mu_{33}
\end{array}  \right):=
\left( \begin{array}{ccc}
E\{W^2\} & E\{WU\} & E\{WY_1\} \\
E\{WU\} & E\{U^2\} & E\{UY_1\}\\
E\{WY_1\} & E\{UY_1\} & E\{Y_1^2\}
\end{array}  \right)
\\
&=\left( \begin{array}{ccc} P_1 & \lambda_1P_1 &
\eta_1 \sqrt{P_1}\\
\lambda_1P_1 & \alpha\beta P_2+\lambda_1^2P_1 &
\lambda_1 \eta_1\sqrt{P_1}+\sqrt{c_{21}}\alpha\beta P_2\\
\eta_1\sqrt{P_1} & \lambda_1
\eta_1\sqrt{P_1}+\sqrt{c_{21}}\alpha\beta P_2 &
\eta_1^2+c_{21}\alpha P_2+1
\end{array}  \right),
\end{align*}
\begin{align*}
{\mathbf{\Sigma}}_{UVY_2} &= \left( \begin{array}{ccc}
\nu_{11} & \nu_{12} & \nu_{13} \\
\nu_{21} & \nu_{22} & \nu_{23}\\
\nu_{31} & \nu_{32} & \nu_{33}
\end{array}  \right):=
\left( \begin{array}{ccc}
E\{U^2\} & E\{UV\} & E\{UY_2\} \\
E\{UV\} & E\{V^2\} & E\{VY_2\}\\
E\{UY_2\} & E\{VY_2\} & E\{Y_2^2\}
\end{array}  \right)
\\
&=\left( \begin{array}{ccc} \alpha\beta P_2 + \lambda_1^2P_1 &
\lambda_1\lambda_2 P_1 &
\alpha\beta P_2 + \lambda_1\eta_2\sqrt{P_1} \\
\lambda_1\lambda_2 P_1 & \alpha\bar{\beta}P_2+\lambda_2^2P_1 &
\alpha\bar{\beta}P_2+\lambda_2\eta_2\sqrt{P_1}\\
\alpha\beta P_2 + \lambda_1\eta_2\sqrt{P_1} &
\alpha\bar{\beta}P_2+\lambda_2\eta_2\sqrt{P_1} & \alpha P_2 +
\eta_2^2 + 1
\end{array}  \right),
\end{align*}
where
\begin{align*}
\eta_1 &= \sqrt{P_1}+\sqrt{c_{21}\bar{\alpha}P_2},\\
\eta_2 &= \sqrt{\bar{\alpha} P_2}+\sqrt{c_{12}P_1},
\end{align*} and $E\{\cdot\}$ denotes the expectation of a random variable.

Define $\Gamma(x) = \log_2(x)/2$, and $\xi = \log_2(2\pi e)/2$. We
express the respective differential entropy terms as:
\begin{align*}
&h_a = h(W) = \xi+\Gamma(\mu_{11}),\\
&h_b = h(UY_1) = 2\xi+\Gamma\left(\left| \begin{array}{cc}
\mu_{22} &
\mu_{23}\\
\mu_{32} & \mu_{33}
\end{array}\right|\right),\\
&h_c = h(WUY_1) = 3\xi+\Gamma\left(\left| \begin{array}{ccc}
\mu_{11} & \mu_{12} & \mu_{13} \\
\mu_{21} & \mu_{22} & \mu_{23}\\
\mu_{31} & \mu_{32} & \mu_{33}
\end{array}  \right|\right),
\\
&h_d = h(UV) = 2\xi+\Gamma\left(\left| \begin{array}{cc} \nu_{11}
&
\nu_{12}\\
\nu_{21} & \nu_{22}
\end{array}\right|\right),\\
&h_e = h(Y_2) = \xi+\Gamma(\nu_{33}),
\\
&h_f = h(UVY_2) = 3\xi+\Gamma\left(\left| \begin{array}{ccc}
\nu_{11} & \nu_{12} & \nu_{13} \\
\nu_{21} & \nu_{22} & \nu_{23}\\
\nu_{31} & \nu_{32} & \nu_{33}
\end{array}  \right|\right),
\\
&h_g = h(WU) = 2\xi+\Gamma\left(\left| \begin{array}{cc} \mu_{11}
&
\mu_{12}\\
\mu_{21} & \mu_{22}
\end{array}\right|\right),\\
&h_h = h(Y_1) = \xi+\Gamma(\mu_{33}),\\
&h_i = h(V) = \xi+\Gamma(\nu_{22}),\\
&h_j = h(UY_2) = 2\xi+\Gamma\left(\left| \begin{array}{cc}
\mu_{11} &
\mu_{13}\\
\mu_{31} & \mu_{33}
\end{array}\right|\right),\\
&h_k = h(U) = \xi+\Gamma(\nu_{11}),\\
&h_l = h(VY_2) = 2\xi+\Gamma\left(\left| \begin{array}{cc}
\mu_{22} &
\mu_{23}\\
\mu_{32} & \mu_{33}
\end{array}\right|\right),
\end{align*}
where $|\cdot|$ denotes the determinant of a matrix.

The mutual information terms in
\eqref{region_ineq1}--\eqref{region_cons_4} are then computed as:
\begin{align*}
&I_1 = h_a + h_b - h_c,\\
&I_2 = h_d + h_e - h_f,\\
&I_3 = \Gamma(1+\frac{\lambda_1^2P_1}{\alpha\beta P_2}),\\
&I_4 = \Gamma(1+\frac{\lambda_2^2P_1}{\alpha\bar{\beta} P_2}),\\
&I_5 = h_g + h_h - h_c,
\\
&I_6 = h_i + h_j - h_f,
\\
&I_7 = h_k + h_l - h_f.
\end{align*}

Let $\mathcal{G}(\alpha,\beta,\lambda_1, \lambda_2)$ denote the
set of all rate pairs $(R_1, R_2)$ such that the following
inequalities are satisfied:
\begin{align}
&R_1 \leq I_1, \label{G_start}\\
&R_2 \leq I_2 - I_3 - I_4, \\
&R_1 + R_2 \leq I_5 + I_6 - I_3 - I_4;
\\
&0 \leq I_5 - I_3,\\
&0 \leq I_7 - I_3,\\
&0 \leq I_6 - I_4,\\
&0 \leq I_2 - I_3 - I_4. \label{G_end}
\end{align}
for given $\alpha,\beta \in [0,1]$ and $\lambda_1, \lambda_2 \in
[0, +\infty)$. Note that \eqref{G_start}--\eqref{G_end} are
directly extended from
\eqref{region_ineq1}--\eqref{region_cons_4}.

\begin{Theorem}
The rate region $\mathcal{G}$ is achievable for the GIC-DMS in the
standard form with
\begin{align*}
\mathcal{G} = \bigcup_{\alpha,\beta \in [0,1];\lambda_1, \lambda_2
\in [0, +\infty)}\mathcal{G}(\alpha,\beta,\lambda_1,\lambda_2).
\end{align*}
\end{Theorem}

\newpage
\bibliographystyle{IEEETran}

\bibliography{IEEEAbrv,D:/MyPapers/mybib/mybib}

\end{document}